\newif\ifconf\conffalse
\ifconf
\documentclass[letterpaper]{sig-alternate}
\else
\documentclass[letterpaper,11pt]{article}
\fi
\usepackage{url}
\ifconf
\else
\usepackage{amsthm}
\usepackage{fullpage}
\fi

\usepackage{amsmath,amssymb,xspace,cite}
\usepackage[usenames]{color}
\usepackage{hyperref}

\DeclareMathOperator{\rank}{rank}
\DeclareMathOperator{\polylog}{polylog}
\DeclareMathOperator{\poly}{poly}

\newcommand{\LM}{\mathop{LM}}
\newcommand{\MR}{\mathop{MR}}
\newcommand{\inprod}[1]{\langle #1 \rangle}
\newcommand{\eqdef}{\mathbin{\stackrel{\rm def}{=}}}
\newcommand{\eps}{\varepsilon}
\newcommand{\tr}{\mathrm{tr}}
\newcommand{\E}{\mathbb{E}\xspace}
\newcommand{\R}{\mathbb{R}}
\renewcommand{\Pr}{\mathbb{P}}

\newcommand{\argmin}{\operatornamewithlimits{argmin}}
\newcommand{\nnz}{\operatornamewithlimits{nnz}}
\newcommand{\SE}{\Pi} %subspace embedding

\newtheorem{theorem}{Theorem}%[section]\numberwithin{equation}{section}
\newtheorem{corollary}[theorem]{Corollary}
\newtheorem{fact}[theorem]{Fact}
\newtheorem{definition}[theorem]{Definition}
\newtheorem{lemma}[theorem]{Lemma}
\newtheorem{remark}[theorem]{Remark}

\newcommand{\proofbelow}{3pt}
\newcommand{\afterproof}{\hfill $\blacksquare$ \par \vspace{\proofbelow}}
\renewenvironment{proof}{\noindent\textbf{Proof.}\,}{\afterproof}

\newcommand{\CorollaryName}[1]{\label{cor:#1}}
\newcommand{\DefinitionName}[1]{\label{def:#1}}
\newcommand{\EquationName}[1]{\label{eq:#1}}
\newcommand{\FactName}[1]{\label{fact:#1}}
\newcommand{\LemmaName}[1]{\label{lem:#1}}

\newcommand{\SectionName}[1]{\label{sec:#1}}
\newcommand{\TheoremName}[1]{\label{thm:#1}}
\newcommand{\FigureName}[1]{\label{fig:#1}}

\newcommand{\Corollary}[1]{Corollary~\ref{cor:#1}}
\newcommand{\Definition}[1]{Definition~\ref{def:#1}}
\newcommand{\Equation}[1]{Eq.\:\eqref{eq:#1}}
\newcommand{\Fact}[1]{Fact~\ref{fact:#1}}
\newcommand{\Lemma}[1]{Lemma~\ref{lem:#1}}

\newcommand{\Section}[1]{Section~\ref{sec:#1}}
\newcommand{\Theorem}[1]{Theorem~\ref{thm:#1}}
\newcommand{\Figure}[1]{Figure~\ref{fig:#1}}

\begin{document}

\author{Jelani Nelson\thanks{Institute for Advanced
    Study. \texttt{minilek@ias.edu}. Supported by NSF
  CCF-0832797 and NSF DMS-1128155.}
\and Huy L. Nguy$\tilde{\hat{\mbox{e}}}$n\thanks{Princeton
  University. \texttt{hlnguyen@princeton.edu}. Supported in part by
  NSF CCF-0832797 and a Gordon Wu fellowship.}}

\title{OSNAP: Faster numerical linear algebra\\ algorithms via sparser subspace embeddings}

\maketitle

\begin{abstract}
An {\it oblivious subspace embedding (OSE)} given some parameters $\eps,d$ is a
distribution $\mathcal{D}$ over matrices
$\SE\in\R^{m\times n}$ such that for any linear subspace
$W\subseteq \R^n$ with $\mathrm{dim}(W) = d$ it holds that
$$\Pr_{\SE\sim \mathcal{D}}(\forall x\in W\ \|\SE x\|_2 \in (1\pm
\eps)\|x\|_2) > 2/3 .$$
We show an OSE exists with $m = O(d^2/\eps^2)$
and where every $\Pi$ in the support of $\mathcal{D}$ has exactly
$s=1$ non-zero entries per column. This improves 
previously best known bound in
[Clarkson-Woodruff, arXiv abs/1207.6365].
Our quadratic dependence on $d$ is optimal for
any OSE with $s=1$
[Nelson-Nguy$\tilde{\hat{\textnormal{e}}}$n, 2012].
We also give two OSE's, which we call Oblivious Sparse Norm-Approximating Projections (OSNAPs), that both allow the parameter settings
$m = \tilde{O}(d/\eps^2)$ and $s
= \mathrm{polylog}(d)/\eps$, or $m = O(d^{1+\gamma}/\eps^2)$ and
$s=O(1/\eps)$ for any constant $\gamma>0$.\footnote{We say $g =
  \tilde{\Omega}(f)$ when $g =
  \Omega(f/ \mathrm{polylog}(f))$, $g = \tilde{O}(f)$ when $g =
  O(f\cdot \mathrm{polylog}(f))$, and $g = \tilde{\Theta}(f)$ when $g
  = \tilde{\Omega}(f)$ and $g = \tilde{O}(f)$ simultaneously.}
This $m$ is nearly optimal since  $m \ge d$ is required simply to
no non-zero vector of $W$ lands in the kernel of $\SE$.
These are the first constructions with $m=o(d^2)$ to have
$s=o(d)$. In fact, our OSNAPs are nothing more than the sparse
Johnson-Lindenstrauss
matrices of [Kane-Nelson, SODA 2012]. Our analyses all yield
OSE's that are sampled using either $O(1)$-wise or $O(\log d)$-wise
independent hash functions, which provides some efficiency
advantages over previous work for turnstile streaming applications.
Our main result is essentially a Bai-Yin type theorem in random matrix
theory and is likely to be of independent interest: i.e.\ we show that
for any
$U\in\R^{n\times d}$ with
orthonormal columns and random sparse $\SE$, all singular values of
$\SE U$ lie in $[1-\eps, 1+\eps]$ with good probability.

Plugging OSNAPs into known
algorithms for numerical linear algebra problems such as approximate
least squares regression, low rank approximation, and approximating
leverage scores implies faster algorithms for all these
problems.
For example, for the approximate least squares
regression problem of computing $x$ that minimizes
$\|Ax - b\|_2$ up to a constant factor, our embeddings
imply a running time of
$\tilde{O}(\nnz(A) + r^{\omega})$, which is
essentially the best bound one could hope for (up to logarithmic
factors).
Here $r = \mathrm{rank}(A)$, $\nnz(\cdot)$ counts non-zero entries,
and $\omega$
is the
exponent of matrix multiplication. Previous
algorithms had a worse dependence on $r$.
\end{abstract}

\section{Introduction}\SectionName{intro}
There has been much recent work on applications of dimensionality
reduction to handling large datasets. Typically special features of
the data such as low ``intrinsic'' dimensionality, or sparsity, are
exploited to reduce the volume of data before processing, thus
speeding up analysis time. One success story of this approach is the
applications of fast algorithms for the Johnson-Lindenstrauss
lemma~\cite{JL84}, which allows one to reduce the dimension of a set
of vectors while preserving all pairwise distances. There have been
two popular lines of work in this area: one focusing on fast
embeddings for all vectors~\cite{AC09,
  AL09,AL11,HV11,KMR12,KW11,Vybiral11},
and one focusing on fast
embeddings specifically for sparse
vectors~\cite{Achlioptas03,BOR10,DKS10,KN10, KN12}.

In this work we focus on the problem of constructing an {\em oblivious
  subspace embedding (OSE)} \cite{Sarlos06} and on
applications of these
embeddings. Roughly speaking, the problem is
to design a data-independent distribution over linear mappings such
that
when data come from an {\em unknown} low-dimensional subspace, they
are reduced to roughly their true dimension while their structure (all
distances in the subspace in this case) is preserved at the same
time. It can be seen as a continuation of the approach based on the
Johnson-Lindenstrauss lemma to subspaces. Here we focus on the setting
of sparse inputs, where it is important that the algorithms take time
proportional to the input sparsity. These embeddings have found
applications in numerical linear algebra problems such as least
squares regression, low rank approximation,
and approximating leverage scores
\cite{CW09,CW12,DMMW12,NDT09,Sarlos06,Tropp11}. We refer the
interested reader to the surveys \cite{HMT11,Mahoney11} for an
overview of this area.

% The organization of the paper is as follows. First the definition of
% the problem and notations are given
% in \Section{prelim}. Next, \Section{analysis} describes two
% constructions of oblivious subspace embeddings. Lastly, applications
% of those embeddings are described in \Section{apps}.

%\section{Preliminaries}\SectionName{prelim}
 Throughout this document we use
$\|\cdot\|$ to
denote $\ell_2$ norm in the case of vector arguments, and
$\ell_{2\rightarrow 2}$ operator norm in the case of matrix arguments.
Recall the definition of the OSE problem.

\begin{definition}\DefinitionName{ose}
The {\em oblivious subspace embedding problem} is to design a
distribution over $m\times n$ matrices $\SE$ such
that for any $d$-dimensional
subspace $W\subset \R^n$, with probability at least $2/3$ over
the choice of $\SE\sim \mathcal{D}$, the
following inequalities hold for all $x\in W$ simultaneously:
$$(1-\eps)\|x\| \le \|\SE x\| \le (1+\eps)\|x\| .$$
Here $n, d, \eps, \delta$ are given parameters of the problem and we
would like $m$ as small as possible.
\end{definition}

OSE's were first introduced in \cite{Sarlos06} as a means to obtain
fast randomized algorithms for several numerical linear algebra
problems. To see the connection, consider for example the least
squares regression problem of computing
$\argmin_{x\in \R^d} \|Ax - b\|$
for some $A\in\R^{n\times d}$. Suppose $\SE\in\R^{m\times n}$ 
 preserves the $\ell_2$ norm up to $1+\eps$ of all vectors in the
subspace spanned by $b$ and the columns of $A$.
Then computing $\argmin_x \|\SE A x - \SE b\|$ instead gives a
solution that
is within $1+\eps$ of optimal. Since the subspace being preserved has
dimension at most $r+1\le d+1$, where $r = \mathrm{rank}(A)$, one only
needs $m = f(r+1,\eps)$ for whatever function $f$ is
achievable in some OSE construction. Thus the running time for
approximate $n\times d$ regression becomes that for
$f(r,\eps)\times d$ regression, plus an additive term for the time
required to compute $\SE A, \SE b$. Even if $A$ has full column
rank and $r=d$ this is still a gain for instances with $n \gg d$.
Also note that the $2/3$ success probability 
guaranteed by \Definition{ose} can be amplified to
$1-\delta$ by running this procedure $O(\log(1/\delta))$ times with
independent randomness and taking the best $x$ found in any run.

Naively there is no gain from the above approach since the time
to compute $\SE A$ could be as large as matrix multiplication between
an $m\times n$ and $n\times d$ matrix. Since $m \ge d$ in any OSE,
this is $O(nd^{\omega - 1})$ time where $\omega < 2.373\ldots$
\cite{Williams12} is the exponent of
square matrix multiplication, and exact least squares regression can
already be computed in this time bound. The work of \cite{Sarlos06}
overcame this barrier by choosing $\SE$ to be a special structured
matrix, with the property that $\SE A$ can be computed in
time $O(nd\log n)$ (see also \cite{Tropp11}). This matrix $\SE$ was
the Fast Johnson-Lindenstrauss Transform of \cite{AC09}, which has the
property that $\SE x$ can be computed in roughly $O(n\log n)$ time for
any $x\in\R^n$. Thus, multiplying $\SE A$ by iterating over columns of
$A$ gives the desired speedup.

The $O(nd\log n)$ running time of the above scheme to compute
$\SE A$
seems almost linear, and thus nearly optimal, since the input
size is already $nd$ to describe $A$. While this is
true for dense $A$, in many practical instances one often expects the
input matrix $A$
to be sparse, in which case linear time in the input description
actually means $O(\nnz(A))$, where $\nnz(\cdot)$ is the number of
non-zero entries. For example consider the case of $A$ being the
Netflix matrix, where $A_{i,j}$ is user $i$'s score for movie $j$:
$A$ is very sparse since most users do not watch, let
alone score, most movies \cite{ZWSP08}. 

In a recent beautiful and surprising work,
\cite{CW12} showed that there exist OSE's with $m =
\mathrm{poly}(d/\eps)$, and where every matrix $\SE$ in the support of
the distribution is {\it very} sparse: even with only $s=1$ non-zero
entries per column! Thus one can transform, for example, an $n\times
d$ least squares regression problem into a $\mathrm{poly(d/\eps)}\times
d$ regression problem in $\nnz(A)$ time. They gave two sparse
OSE constructions: one with $m = \tilde{O}(d^4/\eps^4),s=1$,
and another with $m = \tilde{O}(d^2/\eps^4), s =
O((\log d)/\eps)$.\footnote{Recently after sharing the statement of our bounds with the
authors of \cite{CW12}, independently of our methods they have been
able to push their own methods further to obtain $m =
O((d^2/\eps^2)\log^6(d/\eps))$ with $s=1$, nearly matching our
bound, though only for the $s=1$ case. This improves the two bounds in
the topmost row of \Figure{bounds} under the \cite{CW12} reference
to come within $\polylog d$ or $\polylog k$ factors of the two bounds in our topmost row.} The second
construction is advantageous
when $d$ is
larger as a function of $n$ and one is willing to slightly worsen the
$\nnz(A)$ term in the running time for a gain in the input size of the final
regression problem. 

We also remark that the analyses given of both constructions in
\cite{CW12} require
$\Omega(d)$-wise independent hash functions,
so that from the $O(d)$-wise independent seed used to generate $\SE$
naively one needs an additive $\Omega(d)$ time to identify the
non-zero entries in each
column just to evaluate the hash function. In streaming applications
this can be improved to 
additive $\tilde{O}(\log^2 d)$ time using fast multipoint evaluation
of polynomials (see \cite[Remark 16]{KNPW11}), though ideally if $s = 1$
one could hope for a construction that allows one to find, for any
column, the non-zero entry in that column in constant time given only
a short seed that specifies $\SE$ (i.e.\ without writing down $\SE$
explicitly in memory, which could be
prohibitively expensive for $n$ large in applications such as
streaming and out-of-core numerical linear algebra). Recall that in
the entry-wise
turnstile streaming model, $A$ receives entry-wise updates of the form
$((i,j),v)$, which cause the change $A_{i,j} \leftarrow A_{i,j} +
v$. Updating the embedding thus amounts to adding $v$ times the $j$th
row of $\SE$ to $\SE A$, which should ideally take $O(s)$ time and not
$O(s) + \tilde{O}(\log^2 d)$. 

In the following paragraph we let $S_\SE$ be the
space required to store $\SE$ implicitly (e.g.\ store the seed to
some hash function that specifies $\SE$). We let $t_c$ be
the running time required by an algorithm which, given a column index
and the length-$S_\SE$ seed specifying $\SE$,
returns the list of all non-zeroes in that column in $\SE$. 

\paragraph{Our Main Contribution:} We give an improved analysis of the
$s=1$ OSE in \cite{CW12} and show that it actually achieves
$m=O(d^2/\eps^2), s = 1$. Our analysis is near-optimal since
$m=\Omega(d^2)$ is required for any OSE with $s=1$ \cite{NN12b}.
Furthermore, for this construction we show $t_c = O(1), S_\SE = O(\log
(nd))$.
We also show that the two sparse
Johnson-Lindenstrauss constructions of \cite{KN12} both yield OSE's
that
allow for the parameter settings $m = \tilde{O}(d/\eps^2), s =
\mathrm{polylog}(d)/\eps,t_c = \tilde{O}(s), S_\SE = O(\log d\log
(nd))$ or $m = O(d^{1+\gamma}/\eps^2), s =
O_\gamma(1/\eps), t_c = O((\log d)/\eps), S_\SE = O(\log d\log (nd))$
for any desired constant $\gamma>0$. This $m$ is
nearly optimal since $m\ge d$ is required simply to ensure that no
non-zero vector in the subspace lands in the kernel of $\SE$. Plugging our
improved OSE's into previous work implies faster algorithms for
several numerical linear algebra problems, such as approximate
least squares regression, low rank approximation, and approximating
leverage scores. We remark that both of the OSE's in this work and
\cite{CW12} with $s \gg 1$
have the added benefit of preserving any subspace with $1/\poly(d)$,
and not just constant, failure probability.

\subsection{Problem Statements and Bounds}
\begin{center}
\begin{figure*}
\begin{center}
\begin{footnotesize}
\begin{tabular}{|l|l|l|l|}
\hline
reference & regression & leverage scores & low rank approximation
\\ \hline \cite{CW12} & $O(\nnz(A)) + \tilde{O}(d^5)$&
& $O(\nnz(A))$ + $\tilde{O}(nk^5)$
\\ & $O(\nnz(A)\log n)$ + $\tilde{O}(r^3)$ &
$\tilde{O}(\nnz(A) + r^3)$ & $O(\nnz(A)\log k)$ +
$\tilde{O}(nk^2)$
\\ \hline this work & $O(\nnz(A) + d^3\log d)$ &  & $O(\nnz(A))+\tilde{O}(nk^2)$
\\ & $\tilde{O}(\nnz(A)  + r^\omega)$ &
$\tilde{O}(\nnz(A) + r^{\omega})$ &
$O(\nnz(A)\log^{O(1)} k)+\tilde{O}(nk^{\omega-1})$
\\ & & & $O(\nnz(A))+\tilde{O}(nk^{\omega-1+\gamma})$
\\\hline
\end{tabular}
\end{footnotesize}
\caption{The improvement gained in running times by using our
  OSE's. Dependence on $\eps$ suppressed for readability;
  see \Section{apps} for dependence.}\FigureName{bounds}
\end{center}
\end{figure*}
\end{center}

We now formally define all numerical linear algebra problems we
consider. Plugging our new OSE's into previous algorithms for the
above problems yields the bounds in \Figure{bounds}; the value $r$
used in bounds denotes $\rank(A)$.

\bigskip

\paragraph{Approximating Leverage Scores:} A $d$-dimensional
subspace $W\subseteq \R^n$ can be
written as $W = \{x : \exists y\in\R^d, x = Uy\}$ for
some $U\in\R^{n\times d}$ with orthonormal columns. 
The squared Euclidean norms of rows of $U$ are unique up to
permutation, i.e.\ they depend only on $A$, and are known as the {\it
  leverage scores} of $A$. Given $A$, we
would like to output a list of its leverage scores up to $1\pm\eps$.

\paragraph{Least Squares Regression:} Given $A\in\R^{n\times d}$,
$b\in\R^n$, compute $\tilde{x}\in\R^d$ so that $\|A\tilde{x} - b\| \le
(1+\eps)\cdot \min_{x\in \R^d}\|Ax - b\|$.

\paragraph{Low Rank Approximation:} Given $A\in\R^{n\times d}$ and
integer $k>0$,
compute $\tilde{A}_k\in\R^{n\times d}$ with $\mathrm{rank}(\tilde{A})
\le k$ so that $\|A - \tilde{A}_k\|_F \le (1+\eps) \cdot
\min_{\mathrm{rank}(A_k) \le k} \|A - A_k\|_F$, where $\|\cdot\|_F$
is Frobenius norm.

\subsection{Our Construction and Techniques}
The $s=1$ construction is simply the TZ sketch
\cite{TZ12}. This matrix $\SE$ is specified by a random hash
function $h:[d]\rightarrow[n]$ and a random $\sigma\in\{-1,1\}^d$. For
each $i\in[d]$ we set $\SE_{h(i), i} = \sigma_i$, and every other
entry in $\SE$ is set to zero.
Observe any $d$-dimensional subspace $W\subseteq \R^n$ can be
written as $W = \{x : \exists y\in\R^d, x = Uy\}$ for
some $U\in\R^{n\times d}$ with orthonormal columns. 
The analysis of the $s=1$ construction in \cite{CW12} worked roughly
as follows: let $\mathcal{I}\subset [n]$ denote the set of ``heavy''
rows, i.e.\ those rows $u_i$ of $U$ where $\|u_i\|$ is ``large''. We
write $x =
x_{\mathcal{I}} + x_{[n]\backslash \mathcal{I}}$, where $x_S$ for a set
$S$ denotes $x$ with all coordinates in $[n]\backslash S$ zeroed out.
Then $\|x\|^2 = \|x_{\mathcal{I}}\|^2 +
\|x_{[n]\backslash\mathcal{I}}\|^2 + 2\inprod{x_{\mathcal{I}},
  x_{[n]\backslash\mathcal{I}}}$. The argument in \cite{CW12}
conditioned on $\mathcal{I}$ being perfectly hashed by $h$ so that
$\|x_{\mathcal{I}}\|^2$ is preserved exactly. Using an approach in
\cite{KN10,KN12} based on the Hanson-Wright inequality \cite{HW71}
together with a net argument, it was argued that
$\|x_{[n]\backslash\mathcal{I}}\|^2$ is preserved simultaneously for
all $x\in W$; this step required $\Omega(d)$-wise independence to
union bound over the net. A simpler concentration argument was used to
handle the $\inprod{x_{\mathcal{I}}, x_{[n]\backslash\mathcal{I}}}$
term. The construction in \cite{CW12} with smaller $m$ and larger $s$
followed a similar but more complicated analysis; that construction
involving hashing
into buckets and using the sparse Johnson-Lindenstrauss matrices of
\cite{KN12} in each bucket.

Our analysis is completely
different. First, just as in the TZ sketch's application to $\ell_2$
estimation in data streams, we only require $h$ to be
pairwise independent and $\sigma$ to be $4$-wise independent. Our
observation is simple: a matrix $\SE$
preserving the Euclidean norm of all vectors $x\in W$ up to $1\pm\eps$
is equivalent to the statement $\|\SE U y\| = (1\pm\eps)\|y\|$
simultaneously for all $y\in\R^d$. This is equivalent to 
all singular values of $\SE U$ lying
in the interval $[1-\eps, 1+\eps]$.\footnote{Recall that the singular
  values of a (possibly rectangular) matrix $B$ are the square roots
  of the eigenvalues of $B^*B$, where $(\cdot)^*$ denotes conjugate
  transpose.} Write $S = (\SE U)^* \SE U$, so that we want to show all
eigenvalues values of $S$ lie in $[(1-\eps)^2, (1+\eps)^2]$. We can
trivially write $S = I + (S-I)$, and thus by Weyl's inequality (see a
statement in \Section{analysis}) all eigenvalues of $S$ are $1\pm \|S
- I\|$. We thus show that $\|S-I\|$ is small with good
probability. By Markov's inequality
$$\Pr(\|S-I\|\ge t) = \Pr(\|S-I\|^2 \ge t^2) \le t^{-2}\cdot\E\|S-I\|^2 \le t^{-2}
\cdot\E \|S-I\|_F^2 .$$
Bounding this latter quantity is a simple calculation and fits
in under a page (\Theorem{main}).

The two constructions with smaller $m \approx d/\eps^2$ are the
sparse
Johnson-Lindenstrauss matrices of \cite{KN12}. In particular, the only
properties we need from our OSE in our analyses are the following. Let
each matrix in the support of the OSE have entries in
$\{0,1/\sqrt{s},-1/\sqrt{s}\}$. For a randomly drawn $\SE$, let
$\delta_{i,j}$ be an indicator random variable for the event
$\SE_{i,j} \neq 0$, and write $\SE_{i,j} =
\delta_{i,j}\sigma_{i,j}/\sqrt{s}$, where the $\sigma_{i,j}$ are
random signs. Then the properties we need are
\begin{itemize}
\item For any $j\in [n]$, $\sum_{i=1}^m \delta_{i,j} = s$ with
  probability $1$.
\item For any $S\subseteq [m]\times [n]$, $\E \prod_{(i,j) \in S}
  \delta_{i,j} \le (s/m)^{|S|}$.
\end{itemize}
The second property says the $\delta_{i,j}$ are negatively
correlated.
We call any matrix drawn from an OSE with the above properties
an {\it oblivious sparse norm-approximating projection} (OSNAP).

The work of \cite{KN12} gave two OSNAP distributions, either of which
suffice for our current OSE problem. In the first
construction, each column is chosen to have exactly $s$ non-zero
entries in random locations, each equal to $\pm 1/\sqrt{s}$ uniformly
at random. For our purposes the signs $\sigma_{i,j}$ need only
be $O(\log
d)$-wise independent, and each column can be specified by a $O(\log
d)$-wise
independent permutation, and the seeds specifying the permutations in
different columns need only be $O(\log d)$-wise independent.
In the second construction we pick hash functions
$h:[d]\times [s]\rightarrow[m/s]$, $\sigma:[d]\times
[s]\rightarrow\{-1,1\}$, both $O(\log d)$-wise independent, and thus
each representable using $O(\log d\log nd)$ random bits. For each
$(i,j)\in [d]\times [s]$ we set $\SE_{(j-1)s + h(i,j),i} =
\sigma(i,j)/\sqrt{s}$, and all other entries in $\SE$ are set to zero.
Note also that the TZ sketch is itself an OSNAP with $s=1$.

Just as in the TZ sketch, it suffices to show some tail bound: that
$\Pr(\|S - I\| > \eps')$ is small for some $\eps' = O(\eps)$, where $S
= (\SE U)^*\SE U$. 
Note that if the eigenvalues of $S-I$ are
$\lambda_1,\ldots,\lambda_d$, then the eigenvalues of 
$(S-I)^\ell$ are $\lambda_1^\ell,\ldots,\lambda_d^\ell$. Thus for
$\ell$ even, 
$\tr((S-I)^\ell) = \sum_{i=1}^d \lambda_i^\ell$
is an upper bound on $\|S - I\|^\ell$.
Thus by Markov's inequality with $\ell$ even,
\begin{equation}
\Pr(\|S-I\| \ge t) = \Pr(\|S-I\|^\ell \ge t^\ell) \le t^{-\ell}
\cdot \E \|S - I\|^\ell \le t^{-\ell} \cdot \E \tr((S-I)^\ell) .
\EquationName{trace-strategy}
\end{equation}

Our proof works by expanding the expression $\tr((S-I)^\ell)$ and
computing its expectation. This expression is a sum of exponentially
many monomials, each involving a product of $\ell$ terms. Without
delving into technical details at this point, each such monomial can
be thought of as being in correspondence with some undirected
multigraph (see the dot product multigraphs in the proof of
\Theorem{main-moments}). We group monomials corresponding to the same
graph, bound the contribution from each graph separately, then sum over
all graphs. Multigraphs whose edges all have even multiplicity turn
out to be easier to handle (\Lemma{evengraphs}). However most graphs
$G$ do not have this property. Informally speaking, the contribution
of a graph turns out to be related to the product over its edges of
the contribution of that edge. Let us informally call this
``contribution'' $F(G)$. Thus if $E'\subset E$ is a subset of
the edges of $G$, we can write $F(G) \le F((G|_{E'})^2)/2 +
F((G|_{E\backslash E'})^2)/2$ by AM-GM, where squaring a multigraph
means duplicating every edge, and $G|_{E'}$ is $G$ with all edges in
$E\backslash E'$ removed. This reduces back to the case of even edge
multiplicities, but unfortunately the bound we desire on $F(G)$
depends exponentially on the number of connected components of
$G$. Thus this step is bad, since if $G$ is connected,
then one of $G|_{E'},G|_{E'\backslash E}$ can have {\it many} connected
components for any choice of $E'$. For example if $G$ is a cycle on $N$
vertices, for $E'$ a single edge
almost every vertex in $G_{E'}$
is in its own connected component, and even if $E'$ is every
odd-indexed edge then the number of components blows up to $N/2$. Our
method to overcome this is to show that any $F(G)$ is bounded by some
$F(G')$ with the property that every connected component of $G'$ has
two edge-disjoint spanning trees. We then put one such spanning tree
into $E'$ for each component, so that $G|_{E\backslash E'}$ and $G|_{E'}$
both have the same number of connected components as $G$.

Our approach follows the classical moment method in random
matrix theory; see \cite[Section 2]{Tao12} or \cite{Vershynin12}
introductions to this area. In particular, our approach is inspired by
one taken by Bai and Yin \cite{BY93}, who in our notation were
concerned with the case $n=d$, $U = I$, $\SE$ dense. Most of the
complications in our proof arise because $U$ is not the identity
matrix, so that rows of $U$ are not orthogonal. For example, in the
case of $U$ having orthogonal rows all graphs $G$ in the last
paragraph have no edges other than self-loops and are trivial to
analyze.

\section{Analysis}\SectionName{analysis}

In this section let the orthonormal columns of  $U\in \R^{n\times d}$
be denoted $u^1,\ldots,u^d$. Recall our goal is to show that all
singular values of $\SE U$ lie in the
interval $[1-\eps, 1+\eps]$ with probability $1-\delta$ over the
choice of $\SE$ as long as $s,m$ are sufficiently large. We assume
$\SE$ is an OSNAP with sparsity $s$.
As in \cite{BY93} we make use of
Weyl's inequality (see a proof in \cite[Section 1.3]{Tao12}).

\begin{theorem}[Weyl's inequality]
Let $M,H,P$ be $n\times n$ Hermitian matrices where $M$ has
eigenvalues $\mu_1\ge \ldots \ge \mu_n$, $H$ has eigenvalues $\nu_1
\ge \ldots \ge \nu_n$, and $P$ has eigenvalues $\rho_1\ge \ldots\ge
\rho_n$. Then $\forall\ 1\le i\le n$, it holds that
$ \nu_i + \rho_n \le \mu_i \le \nu_i + \rho_1 $.
\end{theorem}

Let $S = (\SE U)^*\SE U$. 
Letting 
$I$ be the $d\times d$ identity matrix, Weyl's inequality
with $M = S$, $H = (1+\eps^2)I$, and $P = S-(1+\eps^2)I$ implies that
all the eigenvalues of
$S$ lie in the range $[1 + \eps^2 + \lambda_{min}(P),
1+\eps^2 + \lambda_{max}(P)]
\subseteq [1 + \eps^2 - \|P\|, 1 + \eps^2 + \|P\|]$,
where $\lambda_{min}(M)$ (resp.\ $\lambda_{max}(M)$) is the smallest
(resp.\ largest) eigenvalue of $M$. Since $\|P\| \le \eps^2 + \|S -
I\|$, it thus suffices to show
\begin{equation}
\Pr(\|S-I\| > 2\eps - \eps^2) < \delta ,
\EquationName{main-requirement}
\end{equation}
since $\|P\| \le 2\eps$ implies that
all eigenvalues of $S$ lie
in $[(1-\eps)^2, (1+\eps)^2]$.

Before proceeding with our proofs below, observe that
for all $k,k'$
\allowdisplaybreaks
\begin{align*}
S_{k,k'} &= \frac 1s\sum_{r=1}^m \left(\sum_{i=1}^n
  \delta_{r,i}\sigma_{r,i}u^k_i\right) \left(\sum_{i=1}^n
  \delta_{r,i}\sigma_{r,i}u^{k'}_i \right) \\
 {}& = \frac 1s\sum_{i=1}^n u^k_i u^{k'}_i\cdot \left(\sum_{r=1}^m
  \delta_{r,i}\right) + \frac 1s\sum_{r=1}^m \sum_{i\neq
      j}\delta_{r,i}\delta_{r,j}
    \sigma_{r,i} \sigma_{r,j} u^k_i u^{k'}_j\\
{}& = \inprod{u^k, u^{k'}} + \frac 1s\sum_{r=1}^m \sum_{i\neq
      j}\delta_{r,i}\delta_{r,j}
    \sigma_{r,i} \sigma_{r,j} u^k_i u^{k'}_j
\end{align*}

Noting $\inprod{u^k,u^k} = \|u^k\|^2 = 1$ and $\inprod{u^k, u^{k'}} =
0$ for $k\neq k'$, we have for all $k,k'$ 
\begin{equation}
(S - I)_{k,k'} = \sum_{r=1}^m \sum_{i\neq
      j}\delta_{r,i}\delta_{r,j}
    \sigma_{r,i} \sigma_{r,j} u^k_i u^{k'}_j .\EquationName{inprod}
\end{equation}

\begin{theorem}\TheoremName{main}
For $\SE$ an OSNAP with $s=1$ and $\eps\in(0,1)$,
with probability at least $1-\delta$ all singular
values of $\SE U$ are $1\pm\eps$
as long as $m \ge \delta^{-1}(d^2+d)/(2\eps -\eps^2)^2$,
$\sigma$ is $4$-wise independent, and $h$ is pairwise independent.
\end{theorem}
\begin{proof}
We show \Equation{main-requirement}.
Our approach is to bound $\E \|S - I\|^2$ then use Markov's
inequality. Since
\begin{equation}
\Pr(\|S-I\| > 2\eps - \eps^2) = \Pr(\|S-I\|^2 > (2\eps - \eps^2)^2)
\le (2\eps - \eps^2)^{-2}\cdot \E\|S-I\|^2 \le 
(2\eps - \eps^2)^{-2}\cdot \E\|S-I\|_F^2 ,\EquationName{second-moment}
\end{equation}
we can bound $ \E\|S - I\|_F^2$ to
show \Equation{main-requirement}.
Here $\|\cdot \|_F$ denotes Frobenius norm.

Now we bound $\E \|S - I\|_F^2$.
We first deal with the diagonal terms of $S-I$. By \Equation{inprod},
\begin{align*}
 \E (S-I)_{k,k}^2 & = \sum_{r=1}^m \sum_{i\neq j} \frac{2}{m^2}
(u^k_i)^2(u^k_j)^2\\
{} &\le \frac {2}{m} \cdot \|u^i\|^4\\
{} &= \frac 2m,
\end{align*}
and thus the diagonal terms in total contribute at most $2d/m$ to $\E
\|S-I\|_F^2$. 

We now focus on the off-diagonal terms. 
By \Equation{inprod}, $\E (S-I)_{k,k'}^2$ is equal to
\begin{align*}
\frac{1}{m^2}\sum_{r=1}^m \sum_{i\neq j} \left(
  (u^k_i)^2 (u^{k'}_j)^2 + u^k_i u^{k'}_i u^k_j u^{k'}_j \right)
 & =  \frac{1}{m} \sum_{i\neq j} \left(
  (u^k_i)^2 (u^{k'}_j)^2 + u^k_i u^{k'}_i u^k_j u^{k'}_j \right) .
\end{align*}
Noting $0 = \inprod{u^k, u^{k'}}^2 = \sum_{k=1}^n (u^k_i)^2
(u^{k'}_i)^2 + \sum_{i\neq j} u^k_iu^{k'}_iu^k_ju^{k'}_j$ we have that
$\sum_{i\neq j}u^k_i u^{k'}_i u^k_j u^{k'}_j \le 0$, so
\begin{align*}
 \E (S-I)_{k,k'}^2 & \le \frac{1}{m} \sum_{i\neq j} (u^k_i)^2
 (u^{k'}_j)^2\\
{} & \le \frac{1}{m} \|u^i\|^2 \cdot \|u^j\|^2 \\
{} & = \frac{1}{m} .
\end{align*}

Thus summing over $i\neq j$, 
the total contribution from off-diagonal terms to $\E \|S -
I\|_F^2$ is at most $d(d-1)/m$. Thus in total
$\E \|S - I\|_F^2 \le (d^2+d)/m$,
and so \Equation{second-moment} and our setting of $m$ gives
$$ \Pr\left(\|S - I\| > 2\eps - \eps^2\right) < \frac 1{(2\eps -
  \eps^2)^2}\cdot
\frac{d^2 + d}{m} \le  \delta .$$
\end{proof}

Before proving the next theorem, it is helpful to state a few facts
that we will repeatedly use. Recall that $u^i$ denotes the $i$th
column of $U$, and we will let $u_i$ denote the $i$th row of $U$.

% \begin{lemma}\LemmaName{sumidentity}
% For any sequence $c_1,\ldots,c_n$ with $|c_i| \le 1$ for all $i$,
% $\|\sum_{k=1}^n c_k u_ku_k^*\| \le 1$.
% \end{lemma}
% \begin{proof}
% Define $B = \sum_k c_k u_ku_k^*$. Then $B$ is symmetric, so
% $$ \|B\| = \sup_{\|x\| = 1} |x^* B x|  = \sup_{\|x\| = 1} \left|\sum_{k=1}^n
% c_k \inprod{x, u_k}^2\right| \le  \sup_{\|x\| = 1} \sum_{k=1}^n
% \inprod{x, u_k}^2 = \left\|\sum_{k=1}^n u_k u_k^*\right\|.$$
% Now note
% $$ \left(\sum_{k=1}^n u_ku_k^*\right)_{i,j} = e_i^*\left(\sum_{k=1}^n
%   u_ku_k^*\right)e_j =
% \sum_{k=1}^n (u_k)_i(u_k)_j  = \inprod{u^i, u^j} ,$$
% and this inner product is $1$ for $i=j$ and $0$ otherwise. Thus
% $\sum_k u_ku_k^* = I$, implying $\|B\| \le 1$.
% \end{proof}

\begin{lemma}\LemmaName{sumidentity}
$\sum_{k=1}^n u_ku_k^* =  I$.
\end{lemma}
\begin{proof}
$$ \left(\sum_{k=1}^n u_ku_k^*\right)_{i,j} = e_i^*\left(\sum_{k=1}^n
  u_ku_k^*\right)e_j =
\sum_{k=1}^n (u_k)_i(u_k)_j  = \inprod{u^i, u^j} ,$$
and this inner product is $1$ for $i=j$ and $0$ otherwise.
\end{proof}

\begin{lemma}\LemmaName{rownorm}
For all $i\in [n]$, $\|u_i\| \le 1$.
\end{lemma}
\begin{proof}
We can extend $U$ to some orthogonal matrix $U'\in\R^{n\times n}$ by
appending $n - d$ columns. For the rows $u'_i$ of $U'$ we then have
$\|u_i\| \le \|u'_i\| = 1$.
\end{proof}

\begin{theorem}[{\cite{NashWilliams61,Tutte61}}]\TheoremName{spantrees}
A multigraph $G$ has $k$ edge-disjoint spanning trees iff
$$ |E_P(G)| \ge k(|P| - 1)$$
for every partition $P$ of the vertex set of $G$, where $E_P(G)$ is
the set of edges of $G$ crossing between two different partitions
in $P$.
\end{theorem}

The following corollary is standard, and we will later only need it
for the case $k=2$.

\begin{corollary}\CorollaryName{twotrees}
Let $G$ be a multigraph formed by removing at most $k$ edges from a
multigraph $G'$ that has edge-connectivity at least $2k$. Then $G$ must
have at least $k$
edge-disjoint spanning trees.
\end{corollary}
\begin{proof}
For any partition $P$ of the vertex set,
each partition must have at least $2k$ edges leaving it in $G'$. Thus
the
number of edges crossing partitions must be at least $k|P|$ in $G'$,
and thus at least $k|P| - k$ in $G$. \Theorem{spantrees}
thus implies that $G$ has $k$ edge-disjoint spanning trees.
\end{proof}

\begin{fact}\FactName{opnorm}
For any matrix $B\in \mathbb{C}^{d\times d}$, $\|B\| =
\sup_{\|x\|,\|y\| = 1} x^* B y$.
\end{fact}
\begin{proof}
We have $\sup_{\|x\|,\|y\| = 1}x^* B y \le \|B\|$ since $x^*
By \le \|x\| \cdot \|B\| \cdot \|y\|$. To show that unit norm $x,y$
exist which achieve $\|B\|$, let $B = U\Sigma V^*$ be the singular
value decomposition of $B$.
That is, $U,V$ are unitary and $\Sigma$ is diagonal with entries
$\sigma_1 \ge \sigma_2 \ge \ldots \sigma_d \ge 0$ so that $\|B\| =
\sigma_1$. We can then achieve $x^* B y = \sigma_1$ by letting
$x$ be the first column of $U$ and $y$ be the first column 
of $V$.
\end{proof}

\begin{theorem}\TheoremName{main-moments}
For $\SE$ an OSNAP  with
$s=\Theta(\log^3(d/\delta)/\eps)$ and
$\eps\in(0,1)$,
with probability at least $1-\delta$, all singular
values of $\SE U$ are $1\pm\eps$
as long as $m = \Omega(d\log^8(d/\delta)/\eps^2)$
and $\sigma,h$ are $\Omega(\log(d/\delta))$-wise independent.
\end{theorem}
\begin{proof}
We will again show \Equation{main-requirement}. 
Recall that by \Equation{trace-strategy} we have 
\begin{equation}
\Pr(\|S-I\| \ge t) \le t^{-\ell} \cdot \E \tr((S-I)^\ell)
\EquationName{trace-strategy2}
\end{equation}
 for $\ell$
any even integer. We thus proceed by bounding $\E \tr((S-I)^\ell)$
then applying \Equation{trace-strategy2}.

It is easy to verify by induction on $\ell$ that for any 
$B\in\R^{n\times n}$ and $\ell \ge 1$,
$$ (B^\ell)_{i,j} = \sum_{\substack{t_1,\ldots,t_{\ell+1}\in [n]\\t_1
= i, t_{\ell+1} = j}} \prod_{k=1}^\ell B_{t_k, t_{k+1}} ,\hbox{ and
thus }
 \tr(B^\ell) = \sum_{\substack{t_1,\ldots,t_{\ell+1}\in [n]\\t_1
= t_{\ell+1}}} \prod_{k=1}^\ell B_{t_k, t_{k+1}} .$$
Applying this identity to $B = S-I$ yields
\begin{equation}
 \E \tr((S-I)^\ell) = \frac 1{s^\ell}\cdot
\E \sum_{\substack{k_1,k_2,\ldots,k_{\ell+1}\\k_1 = k_{\ell+1}\\i_1
  \neq j_1,\ldots, i_\ell\neq j_\ell\\r_1,\ldots,r_\ell}} 
\prod_{t=1}^\ell \delta_{r_t,i_t} \delta_{r_t, j_t} \sigma_{r_t, i_t}
\sigma_{r_t, j_t}
u^{k_t}_{i_t} u^{k_{t+1}}_{j_t}.\EquationName{bigsum}
\end{equation}

The general strategy to bound the above summation is the
following. Let $\Psi$ be the set of all monomials appearing on the
right hand side of \Equation{bigsum}. 
For $\psi\in \Psi$ define $K(\psi) = (k_1,\ldots,k_\ell)$  as the ordered
tuple of $k_t$ values in $\psi$, and similarly define $P(\psi) =
((i_1,j_1),\ldots,(i_\ell,j_\ell))$ and $W(\psi) =
(r_1,\ldots,r_\ell)$. For each $\psi \in
\Psi$ we associate a three-layered undirected multigraph $G_\psi$ with
labeled edges and unlabeled vertices. We call
these three layers the {\it left}, {\it middle}, and {\it right}
layers, and we refer to vertices in the left layer as {\it
  left vertices}, and similarly for vertices in the other
layers. Define $M(\psi)$ to be the set
$\{i_1,\ldots,i_\ell,j_1,\ldots,j_\ell\}$ and define $R(\psi) =
\{r_1,\ldots,r_\ell\}$. We define $y = |M(\psi)|$ and $z =
|R(\psi)|$. Note it can happen that $y < 2\ell$ if some $i_t =
i_{t'}$, $j_t = j_{t'}$, or $i_t = j_{t'}$, and similarly we may also
have $z < \ell$. The graph $G_\psi$ has
$x=\ell$ left vertices, $y$ middle vertices
corresponding to the distinct $i_t,j_t$ in $\psi$,
and $z$ right
vertices corresponding to the distinct $r_t$. For the sake of brevity,
often we refer to the vertex
corresponding to $i_t$ (resp.\ $j_t,r_t$) as simply $i_t$ (resp.\
$j_t,
r_t$). Thus note that when we refer to for example some vertex
$i_t$, it may happen that some other $i_{t'}$ or $j_{t'}$ is also
the same vertex. We now describe the edges of $G_\psi$. For
$\psi = \prod_{t=1}^\ell \delta_{r_t,i_t} \delta_{r_t, j_t}
\sigma_{r_t, i_t} \sigma_{r_t, j_t} u^{k_t}_{i_t} u^{k_{t+1}}_{j_t}$
we draw $4\ell$ labeled edges in $G_{\psi}$ with distinct labels in
$[4\ell]$. For each $t\in [\ell]$ we draw an edge from the $t$th left
vertex to $i_t$
with label $4(t-1)+1$, from $i_t$ to $r_t$ with label $4(t-1) + 2$,
from $r_t$ to $j_t$ with label $4(t-1) + 3$, and from $j_t$ to
the $(t+1)$st left vertex with label $4(t-1) + 4$. Observe that many
different monomials $\psi$ will map to the same graph $G_{\psi}$; in
particular the graph maintains no information concerning equalities
amongst the $k_t$, and
the $y$ middle vertices may map to any $y$ distinct values in
$[n]$ (and similarly the right vertices may map to any $z$
distinct values in $[m]$). We handle the right
hand side of \Equation{bigsum} by grouping monomials $\psi$ that map
to the same graph, bound the total contribution of a given graph
$G$ in terms of its graph structure when summing over all $\psi$ with
$G_\psi = G$, then sum the contributions from all such graphs $G$
combined.

Before continuing further we introduce some more notation then make a
few
observations. For a graph $G$ as above, recall $G$ has
$4\ell$ edges, and we refer
to the {\it distinct} edges (ignoring labels) as {\it bonds}. 
We let $E(G)$ denote the edge multiset of a multigraph $G$ and $B(G)$ denote the
bond set. We refer to the number of bonds a vertex is incident upon as
its {\it bond-degree}, and the number of edges as its {\it
  edge-degree}. We do not count self-loops for calculating
bond-degree, and we count them twice for edge-degree.
We let $\LM(G)$ be the induced multigraph on the left and middle
vertices of $G$, and $\MR(G)$ be the induced multigraph on the middle and
right vertices. 
We let $w = w(G)$
be the number of connected components in $\MR(G)$. 
We let
$b = b(G)$ denote the number of bonds in $\MR(G)$ (note $\MR(G)$ has
$2\ell$
edges, but it may happen that $b<2\ell$ since $G$ is a multigraph).
Given $G$ we define
the undirected {\it dot product multigraph} $\widehat{G}$ with vertex
set $M(\psi)$.
Note every left vertex of $G$ has edge-degree $2$. 
For each $t\in [\ell]$ an edge $(i,j)$ is drawn in $\widehat{G}$
between the two
middle vertices that the $t$th left vertex is adjacent to (we draw a
self-loop on $i$ if
$i=j$). We do not label the edges of $\widehat{G}$, but we label the
vertices with distinct labels in $[y]$ in increasing order of when
each vertex was first visited by the natural tour of $G$ (by following
edges in increasing label order).
We name $\widehat{G}$ the dot product multigraph since if some left
vertex $t$ has its two edges connecting to vertices $i,j\in[n]$,
then summing over $k_t\in[d]$
produces the dot product $\inprod{u_i, u_j}$. 

Now we make some observations. Due to
the random signs $\sigma_{r,i}$, a monomial
$\psi$ has expectation zero unless every bond in $\MR(G)$ has even
multiplicity, in which case the product of
random signs in $\psi$ is $1$. Also, note 
the expectation of the product of
the $\delta_{r,i}$ terms in $\psi$ is at most $(s/m)^b$ by OSNAP
properties. Thus letting $\mathcal{G}$ be the set
of all such graphs $G$ with even bond multiplicity in $\MR(G)$ that
arise from some monomial $\psi$ appearing in \Equation{bigsum}, we
have
\begin{align}
\nonumber \E \tr((S-I)^\ell) &\le \frac 1{s^{\ell}}\cdot\sum_{G\in\mathcal{G}}
\left(\frac sm\right)^b \cdot \left|\sum_{\psi:G_\psi = G} \prod_{t=1}^\ell
u^{k_t}_{i_t} u^{k_{t+1}}_{j_t}\right|\\
\nonumber {} &= \frac 1{s^{\ell}}\cdot\sum_{G\in\mathcal{G}}
\left(\frac sm\right)^b \binom{m}{z} \cdot
\left|\sum_{\substack{\psi:G_\psi = G\\R(\psi)=[z]}} \sum_{k_1,\ldots,k_\ell}\prod_{t=1}^\ell
u^{k_t}_{i_t} u^{k_{t+1}}_{j_t}\right|\\
{} &= \frac 1{s^{\ell}}\cdot\sum_{G\in\mathcal{G}}
\left(\frac sm\right)^b \binom{m}{z} \cdot
\left|\sum_{\substack{a_1,\ldots,a_y\in[n]\\\forall i\neq j\ a_i\neq a_j}}
\prod_{\substack{e\in E(\widehat{G})\\e=(i,j)}}
\inprod{u_{a_i}, u_{a_j}}\right| \EquationName{last-step}
\end{align}

Before continuing further it will be convenient to introduce a notion
we will use in our analysis called a {\it generalized dot product
  multigraph}. Such a graph $\widehat{G}$ is just as in the case of a
dot product multigraph, except that  each
edge $e = (i,j)$ is associated with
some matrix $M_e$. 
We call $M_e$ the {\it edge-matrix} of $e$.
Also since $\widehat{G}$ is undirected, we can think of
an edge $e = (i,j)$ with edge-matrix $M_e$ also as an edge $(j,i)$, in
which case we say its associated edge-matrix is $M_e^*$.
 We then
associate with $\widehat{G}$ the product
$$\prod_{\substack{e\in\widehat{G}\\e=(i,j)}} \inprod{u_{a_i}, M_e u_{a_j}} .$$
Note that a dot product multigraph is simply a generalized
dot product multigraph in which $M_e = I$ for all $e$. Also, in such a
generalized dot product multigraph, we treat multiedges as
representing the same bond iff the associated edge-matrices are also
equal (in general multiedges may have different edge-matrices).

\begin{lemma}\LemmaName{evengraphs}
Let $H$ be a connected generalized dot product multigraph on
vertex set $[N]$ with $E(H)\neq \emptyset$ and where every bond has even
multiplicity. Also
suppose that for all $e\in E(H)$, $\|M_e\| \le 1$. Define
$$ f(H) = \sum_{a_2=1}^n \cdots \sum_{a_N=1}^n
\prod_{\substack{e\in E(H)\\e = (i,j)}} \inprod{v_{a_i}, M_e
  v_{a_j}} ,$$
where $v_{a_i} = u_{a_i}$ for $i\neq 1$, and $v_{a_1}$ equals some fixed
vector $c$ with $\|c\| \le 1$. Then $f(H) \le
\|c\|^2$.
\end{lemma}
\begin{proof}
Let $\pi$ be some permutation of $\{2,\ldots,N\}$. For a bond $q =
(i,j)\in B(H)$, let $2\alpha_q$ denote the multiplicity of $q$ in $H$.
Then by ordering the assignments of the $a_t$ in the summation
$$
\sum_{a_2,\ldots,a_N\in[n]} \prod_{\substack{e\in E(H)\\e=(i,j)}}
\inprod{v_{a_i}, M_e v_{a_j}}
$$
according to $\pi$, we obtain the exactly equal expression
\begin{equation}
 \sum_{a_{\pi(N)} = 1}^n \prod_{\substack{q\in
    B(H)\\ q = (\pi(N),j)\\ N\le \pi^{-1}(j)}}
  \inprod{v_{a_{\pi(N)}}, M_q v_{a_j}}^{2\alpha_q}\cdots \sum_{a_{\pi(2)}
    = 1}^n \prod_{\substack{q\in
    B(H)\\ q = (\pi(1),j)\\ 2\le \pi^{-1}(j)}}
  \inprod{v_{a_{\pi(2)}}, M_q v_{a_j}}^{2\alpha_q}
  .\EquationName{ordersum}
\end{equation}
Here we have taken the product over $t \le \pi^{-1}(j)$ as opposed to
$t<\pi^{-1}(j)$ since there may be self-loops. By \Lemma{rownorm} and
the fact that $\|c\| \le 1$ we
have that for any $i,j$, $\inprod{v_i,v_j}^2 \le \|v_i\|^2\cdot
\|v_j\|^2 \le 1$, so we obtain an upper bound on \Equation{ordersum}
by replacing each $\inprod{v_{a_{\pi(t)}}, v_{a_j}}^{2\alpha_v}$ term
with $\inprod{v_{a_{\pi(t)}}, v_{a_j}}^2$. We can thus obtain the
sum
\begin{equation}
 \sum_{a_{\pi(N)} = 1}^n \prod_{\substack{q\in
    B(H)\\ q = (\pi(N),j)\\ q\le \pi^{-1}(j)}}
  \inprod{v_{a_{\pi(N)}}, M_q v_{a_j}}^2\cdots \sum_{a_{\pi(2)}
    = 1}^n \prod_{\substack{q\in
    B(H)\\ q = (\pi(2),j)\\ 2\le \pi^{-1}(j)}}
  \inprod{v_{a_{\pi(2)}}, M_q v_{a_j}}^2
  ,\EquationName{ordersum2}
\end{equation}
which upper bounds \Equation{ordersum}. Now note for $2\le t\le N$
that for any
nonnegative integer
$\beta_t$ and for $\{q\in B(H):
q=(\pi(t), j), t< \pi^{-1}(j)\}$ non-empty (note the strict inequality
$t<\pi^{-1}(j)$),
\begin{align}
\sum_{a_{\pi(t)}
    = 1}^n \|v_{a_{\pi(t)}}\|^{2\beta_t}\cdot \prod_{\substack{q\in
    B(H)\\ q = (\pi(t),j)\\ t\le \pi^{-1}(j)}}
  \inprod{v_{a_{\pi(t)}}, M_q v_{a_j}}^2 
&\le\sum_{a_{\pi(t)}
    = 1}^n \prod_{\substack{q\in
    B(H)\\ q = (\pi(t),j)\\ t\le \pi^{-1}(j)}}
  \inprod{v_{a_{\pi(t)}}, M_q v_{a_j}}^2 \EquationName{smallrownorm}\\
\nonumber {}&\le \prod_{\substack{q\in
    B(H)\\ q = (\pi(t),j)\\ t< \pi^{-1}(j)}}
  \left( \sum_{a_{\pi(t)}
    = 1}^n \inprod{v_{a_{\pi(t)}}, M_qv_{a_j}}^2\right)\\
\nonumber {}&= \prod_{\substack{q\in
    B(H)\\ q = (\pi(t),j)\\ t< \pi^{-1}(j)}}
  \left( \sum_{a_{\pi(t)}
    = 1}^n v_{a_j}^*M_q^*v_{a_{\pi(t)}}v_{a_{\pi(t)}}^*M_qv_{a_j}\right)\\
\nonumber {}& = \prod_{\substack{q\in
    B(H)\\ q = (\pi(t),j)\\ t< \pi^{-1}(j)}}
  (M_q v_{a_j})^*\left( \sum_{i=1}^nu_iu_i^*\right)M_qv_{a_j}\\
{}& = \prod_{\substack{q\in
    B(H)\\ q = (\pi(t),j)\\ t< \pi^{-1}(j)}}
  \|M_q v_{a_j}\|^2\EquationName{useident} \\
{}& \le \prod_{\substack{q\in
    B(H)\\ q = (\pi(t),j)\\ t< \pi^{-1}(j)}}
  \|v_{a_j}\|^2\EquationName{bounded-op} ,
\end{align}
where \Equation{smallrownorm} used \Lemma{rownorm},
\Equation{useident} used \Lemma{sumidentity}, and
\Equation{bounded-op} used that $\|M_q\| \le 1$. Now consider processing
the alternating sum-product in \Equation{ordersum2} from right to
left. We say that a bond $(i,j)\in B(H)$ is {\it assigned to
  $i$} if $\pi^{-1}(i) < \pi^{-1}(j)$.
When arriving at the $t$th sum-product and using the upper bound
\Equation{useident} on the previous $t-1$ sum-products, we will have
a sum over $\|v_{a_{\pi(t)}}\|^2$
raised to some nonnegative power (specifically the number of bonds
incident upon $\pi(t)$ but not assigned to $\pi(t)$, plus one if
$\pi(t)$ has a
self-loop) multiplied by a product of $\inprod{v_{a_{\pi(t)}}, v_{a_j}}^2$
over all bonds $(\pi(t), j)$ assigned to $\pi(t)$. There are two
cases. In the first case $\pi(t)$ has no bonds assigned to it. 
We will ignore this case since we will show that we can choose $\pi$
to avoid it.
% In this
% case for some $\beta_t > 0$ we must have a sum over
% $\|v_{a_{\pi(t)}}\|^{2\beta_t} \le \|v_{a_{\pi(t)}}\|^2$ (by
% \Lemma{rownorm} and the fact that $\|c\| \le 1$); we must have
% $\beta_t > 0$ since $\pi(t)$ has positive edge-degree in $H$ since $H$
% is connected and $E(H) \neq \emptyset$. We then
% have the upper bound
% $$ \sum_{a_{\pi(t)}=1}^n \|v_{a_{\pi}(t)}\|^2 = \sum_{a_{\pi(t)}=1}^n
% \|u_{a_{\pi}(t)}\|^2 = \|U\|_F^2 = d . $$

The other case is that $\pi(t)$ has at least one bond assigned to
it. In this case we are in the scenario of \Equation{useident} and
thus summing over $a_{\pi(t)}$ yields a non-empty product of 
$\|v_{a_j}\|^2$ for the $j$ for which $(\pi(t), j)$ is a bond assigned
to $\pi(t)$. Thus in our final sum, as long as we choose $\pi$ to
avoid the first case, we are left with an upper bound of $\|c\|$
raised to some power equal to the edge-degree of vertex $1$ in $H$,
which is at least $2$. The lemma would then follow since $\|c\|^j \le
\|c\|^2$ for $j \ge 2$.

It now remains to show that we can choose $\pi$ to avoid the first
case where some $t\in \{2,\ldots,N\}$ is such that $\pi(t)$ has no
bonds assigned to it. Let $T$ be a spanning tree in $H$ rooted at
vertex $1$. We then choose any $\pi$ with the property that for any
$i<j$, $\pi(i)$ is not an ancestor of $\pi(j)$ in $T$. This can be
achieved, for example, by assigning $\pi$ values in reverse breadth
first search order.
\end{proof}

\begin{lemma}\LemmaName{amgm}
Let $\widehat{G}$ be any dot product graph as in \Equation{last-step}.
Then
$$
\left|\sum_{\substack{a_1,\ldots,a_y\in[n]\\\forall i\neq j\ a_i\neq a_j}}
\prod_{\substack{e\in \widehat{G}\\e=(i,j)}}
\inprod{u_{a_i}, u_{a_j}}\right| \le y! \cdot d^{y - w + 1} .$$
\end{lemma}
\begin{proof}
We first note that
we have the inequality
\begin{align*}
 \left|\sum_{\substack{a_1,\ldots,a_y\in[n]\\\forall i\neq j\ a_i\neq a_j}}
\prod_{\substack{e\in E(\widehat{G})\\e=(i,j)}}
\inprod{u_{a_i}, u_{a_j}}\right| &= \left|
\sum_{\substack{a_1,\ldots,a_{y-1}\in[n]\\\forall i\neq j\in[y-1]\
    a_i\neq a_j}} \left(\sum_{a_y=1}^n
\prod_{\substack{e\in E(\widehat{G})\\e=(i,j)}}
\inprod{u_{a_i}, u_{a_j}} - \sum_{t=1}^{y-1}\sum_{a_y = a_t}
\prod_{\substack{e\in E(\widehat{G})\\e=(i,j)}}
\inprod{u_{a_i}, u_{a_j}}\right)\right|\\
{} &\hspace{-.85in}\le \left|
\sum_{\substack{a_1,\ldots,a_{y-1}\in[n]\\\forall i\neq j\in[y-1]\
    a_i\neq a_j}} \sum_{a_y=1}^n
\prod_{\substack{e\in E(\widehat{G})\\e=(i,j)}}
\inprod{u_{a_i}, u_{a_j}}\right| +
\sum_{t=1}^{y-1}\left|\sum_{\substack{a_1,\ldots,a_{y-1}\in[n]\\\forall
      i\neq j\in[y-1]\ a_i\neq a_j}}\sum_{a_y =  a_t}
\prod_{\substack{e\in E(\widehat{G})\\e=(i,j)}}
\inprod{u_{a_i}, u_{a_j}}\right|
\end{align*}
We can view the sum over $t$ on the right hand side of the above as
creating $t-1$ new dot product multigraphs, each with one fewer vertex
where we eliminated vertex $y$ and associated it with vertex $t$ for
some $t$, and for each edge $(y,a)$ we effectively replaced it with
$(t,a)$. Also in first sum where we sum over all $n$ values of $a_y$,
we have eliminated the constraints $a_y \neq a_i$ for $i\neq y$. By
recursively applying this inequality to each of the resulting $t$
summations, we bound 
$$ \left|\sum_{\substack{a_1,\ldots,a_y\in[n]\\\forall i\neq j\
      a_i\neq a_j}}
\prod_{\substack{e\in E(\widehat{G})\\e=(i,j)}}
\inprod{u_{a_i}, u_{a_j}}\right|$$
by a sum of contributions from $y!$ dot product multigraphs where in
none of these multigraphs do we have the constraint that $a_i\neq a_j$
for $i\neq j$. We will show that each one of these resulting
multigraphs contributes at most $d^{y-w+1}$, from which the lemma
follows.

Let $G'$ be one of the dot product multigraphs at a leaf of the above
recursion so that we now wish to bound
\begin{equation}
F(G') \eqdef \left|\sum_{a_1,\ldots,a_y=1}^n
\prod_{\substack{e\in E(\widehat{G'})\\e=(i,j)}}
\inprod{u_{a_i}, M_e u_{a_j}}\right| \EquationName{leafgraph}
\end{equation}
where $M_e = I$ for all $e$ for $G'$.
Before proceeding, we first claim that every connected component of
$G'$ is Eulerian. To see this, observe $G$
has an Eulerian tour, by following the edges of $G$ in
increasing order of label, and thus all middle vertices have even
edge-degree in $G$. However they also have even edge-degree in
$\MR(G)$, and thus the
edge-degree of a middle vertex in $\LM(G)$ must be even as well. Thus, 
every vertex in $\widehat{G}$ has even edge-degree, and thus
every vertex in each of the recursively created leaf graphs also has even
edge-degree since at every step when we eliminate a vertex, some other
vertex's degree increases by the eliminated vertex's degree which was
even. Thus every connected component of $G'$ is Eulerian as desired.

We now upper bound $F(G')$. Let the
connected components of $G'$ be $C_1,\ldots,C_{CC(G')}$, where
$CC(\cdot)$ counts connected
 components. An observation we repeatedly use later is that for any
 generalized dot product multigraph $H$ with components
 $C_1,\ldots,C_{CC(H)}$,
\begin{equation}
F(H) = \prod_{i=1}^{CC(H)} F(C_i) .
\EquationName{component-factors}
\end{equation}
We treat $G'$ as a generalized dot
product multigraph so that each edge $e$ has an associated matrix
$M_e$ (though in fact $M_e = I$ for all $e$). 
% Since $G'$ is Eulerian, every cut of some $C_j$ has an even number of
% edges crossing it which is at least $2$. 
Define an undirected multigraph to be {\it good} if all its connected
components have two edge-disjoint spanning trees.
We will show that $F(G') \le
F(G'')$ for some generalized dot product multigraph $G''$ that is
good then will show
$F(G'') \le d^{y-w+1}$. If $G'$ itself is good then we
can set $G'' = G'$. Otherwise, 
we
will show $F(G') = F(H_0) = \ldots = F(H_\tau)$ for smaller and
smaller generalized dot
product multigraphs $H_t$ (i.e.\ with successively
fewer vertices) whilst maintaining the invariant that 
each $H_t$ has Eulerian connected components and has $\|M_e\| \le
1$ for all $e$. We stop when some $H_\tau$
is good and we can set $G'' = H_\tau$.

Let us now focus on constructing this sequence of $H_t$ in the case
that $G'$ is not good. Let $H_0 = G'$. Suppose we have
constructed $H_0,\ldots,H_{t-1}$ for $i\ge 1$ none of which are
good, and now we want to
construct $H_t$.  Since $H_{t-1}$ is not good it cannot be
$4$-edge-connected by \Corollary{twotrees}, so there is
some connected component $C_{j^*}$ of $H_{t-1}$ with some cut $S\subsetneq
V(C_{j^*})$ with $2$ edges crossing the cut $(S, V(C_{j^*})\backslash
S)$ (note that since $C_{j^*}$ is Eulerian, any cut has an even number
of edges crossing it).
Choose such an $S\subsetneq V(C_{j^*})$ with $|S|$ minimum amongst
all such cuts. Let the two edges crossing the cut be $(g,h),
(g',h')$ with $h,h'\in S$ (note that it may be the case that $g=g'$
and/or $h=h'$). Note that $F(C_{j^*})$ equals the magnitude of
{\footnotesize \begin{equation}
\sum_{a_{V(C_{j^*})\backslash S}\in [n]^{|V(C_{j^*})\backslash
    S|}}\left(\prod_{\substack{e\in
    E(V(C_{j^*})\backslash S)\\e=(i,j)}}
 \inprod{u_{a_i}, M_e u_{a_j}}\right)
u_{a_g}^*\underbrace{M_{(g,h)}\left(\sum_{a_S\in [n]^{|S|}}u_{a_h}
\left(\prod_{\substack{e\in
    E(C_{j^*}(S))\\e=(i,j)}}
\inprod{u_{a_i}, M_e u_{a_j}}\right) u_{a_{h'}}^*\right)M_{(h',g')}}_{M}
u_{a_{g'}} .
\EquationName{hcut}
\end{equation}}
We define $H_t$ to be $H_{t-1}$ but where in the $j^*$th component we
replace $C_{j^*}$ with $C_j^*(V(C_j^*)\backslash S)$ and add an
additional edge
from $g$ to $g'$ which we assign edge-matrix $M$. We thus have
that $F(H_{t-1}) = F(H_t)$ by
\Equation{component-factors}. Furthermore each component of $H_t$ is
still Eulerian since every vertex in $H_{t-1}$ has either been
eliminated, or its edge-degree has been preserved and thus all
edge-degrees are even. It remains to show that $\|M\| \le 1$.

We first claim that
$C_{j^*}(S)$ has two edge-disjoint spanning trees. 
Define $C'$ to be the graph $C_{j^*}(S)$ with an edge from $h$ to $h'$
added. We show that $C'(S)$ is $4$-edge-connected so that $C_{j^*}(S)$
has two edge-disjoint spanning trees by \Corollary{twotrees}.
Now to see this, consider some $S'\subsetneq S$.
Consider the cut $(S', V(C')\backslash S')$.
$C'$ is Eulerian, so the  number of edges crossing this cut is either
$2$ or at least $4$. If it $2$, then since $|S'| < |S|$ this is a
contradiction since $S$ was chosen amongst such cuts to have $|S|$
minimum. Thus it is at least $4$, and we claim that the number of edges
crossing the cut $(S', S\backslash S')$ in $C'(S)$ must also be at
least $4$. If not, then it is $2$ since $C'(S)$ is Eulerian. However
since the number of edges leaving $S'$ in $C'$ is at least $4$, it
must then be
that $h,h'\in S'$. But then the cut $(S\backslash S', V(C')\backslash
(S\backslash S'))$ has $2$ edges crossing it so that $S\backslash S'$
is a smaller cut than $S$ with $2$ edges leaving it in $C'$, violating
the
minimality of $|S|$, a contradiction. Thus $C'(S)$ is
$4$-edge-connected, implying $C_{j^*}(S)$ has two edge-disjoint
spanning trees $T_1,T_2$ as desired.

Now to show $\|M\| \le 1$,  by \Fact{opnorm} we have $\|M\| =
\sup_{\|x\|,\|x'\| = 1} x^* M x'$. We have that
\begin{align}
\nonumber x^* M x' &= \sum_{a_S\in [n]^{|S|}}\inprod{x,M_{(g,h)}u_{a_h}}
\cdot \left(\prod_{\substack{e\in
    E(C_{j^*}(S))\\e=(i,j)}}
\inprod{u_{a_i}, M_e u_{a_j}}\right)\cdot  \inprod{u_{a_{h'}},M_{(h',g')}
x'}\\
\nonumber {}&= \sum_{a_S\in [n]^{|S|}} \left( \inprod{x,M_{(g,h)}u_{a_h}}\cdot
  \prod_{\substack{e\in T_1\\e=(i,j)}} \inprod{u_{a_i}, M_e
    u_{a_j}}\right) \cdot \left(\inprod{u_{a_{h'}},M_{(h',g')} x'}
\cdot \prod_{\substack{e\in E(C_{j^*}(S))\backslash T_1\\e=(i,j)}}
\inprod{u_{a_i}, M_e u_{a_j}} \right)\\
\nonumber {}&\le \frac 12 \cdot \left [\sum_{a_S\in [n]^{|S|}} \left(
    \inprod{x,M_{(g,h)}u_{a_h}}^2\cdot
  \prod_{\substack{e\in T_1\\e=(i,j)}} \inprod{u_{a_i}, M_e
    u_{a_j}}^2\right)\right.\\
{}& \hspace{1in}\left. + \sum_{a_S\in [n]^{|S|}}
\left(\inprod{u_{a_{h'}},M_{(h',g')} x'}^2
\cdot \prod_{\substack{e\in E(C_{j^*}(S))\backslash T_1\\e=(i,j)}}
\inprod{u_{a_i}, M_e u_{a_j}}^2 \right) \right ] \EquationName{amgm}\\
{}& \le \frac 12 \left(\|x\|^2 + \|x'\|^2\right)
\EquationName{useeven} \\
\nonumber {}& = 1 ,
\end{align}
where \Equation{amgm} used the AM-GM inequality, and
\Equation{useeven} used
\Lemma{evengraphs} (note the graph with vertex set $S\cup \{g'\}$ and
edge set
$E(C_{j^*}(S))\backslash T_1 \cup \{(g',h')\}$ is connected since $T_2
\subseteq E(C_{j^*}(S))\backslash T_1$). Thus we have shown that $H_t$
satisfies the desired properties. Now notice that the sequence
$H_0,\ldots,H_1,\ldots$\ must eventually terminate since the number of
vertices is strictly decreasing in this sequence and any Eulerian
graph on $2$ vertices is good. Therefore we have that $H_{\tau}$ is
eventually good for some $\tau > 0$ and we can set $G'' = H_{\tau}$.

It remains to show that for our final good $G''$ we have $F(G'') \le
d^{y-w+1}$. We will show this in two parts by showing that both
$CC(G'') \le d^{y-w+1}$ and $F(G'') \le d^{CC(G'')}$. For the first
claim, note that $CC(G'') \le CC(\widehat{G})$ since every $H_t$ has
the same number of connected components as $G'$, and $CC(G') \le
CC(\widehat{G})$. This latter inequality holds since in each level of
recursion used to eventually obtain $G'$ from $\widehat{G}$, we
repeatedly identified two vertices as equal and merged them, which can
only decrease the number of connected components. 
Now, all
middle vertices in $G$ lie in one connected component (since $G$ is
connected)
 and $\MR(G)$ has $w$ connected components. Thus the at least $w-1$
edges connecting these components in $G$ must come from $\LM(G)$,
implying that $\LM(G)$ (and thus $\widehat{G}$) has at most $y-w+1$
 connected components, which thus must also be true for $G''$ as
 argued above.

It only remains to show $F(G'') \le d^{CC(G'')}$. Let $G''$ have
connected components $C_1,\ldots,C_{CC(G'')}$ with each $C_j$ having
$2$ edge-disjoint spanning trees $T^j_1,T^j_2$. We then have
\begin{align}
\nonumber F(G'') &= \prod_{t=1}^{CC(G'')} F(C_t)\\
\nonumber {} &= \prod_{t=1}^{CC(G'')} \left|\sum_{a_1,\ldots,a_{|V(C_t)|}=1}^n
\prod_{\substack{e\in E(C_t)\\e=(i,j)}}
\inprod{u_{a_i}, M_e u_{a_j}}\right|\\
\nonumber {} &= \prod_{t=1}^{CC(G'')} \left|\sum_{a_1,\ldots,a_{|V(C_t)|}=1}^n
\left(\prod_{\substack{e\in T^t_1\\e=(i,j)}}
\inprod{u_{a_i}, M_e u_{a_j}}\right)\cdot \left(\prod_{\substack{e\in
    E(C_t)\backslash T^t_1\\e=(i,j)}}
\inprod{u_{a_i}, M_e u_{a_j}}\right)\right|\\
 {} &\le \prod_{t=1}^{CC(G'')}\frac 12\left[
  \sum_{a_1=1}^n\sum_{a_2,\ldots,a_{|V(C_t)|}=1}^n
\prod_{\substack{e\in T^t_1\\e=(i,j)}}
\inprod{u_{a_i}, M_e u_{a_j}}^2 +
\sum_{a_1=1}^n\sum_{a_2,\ldots,a_{|V(C_t)|}=1}^n\prod_{\substack{e\in
    E(C_t)\backslash T^t_1\\e=(i,j)}}
\inprod{u_{a_i}, M_e u_{a_j}}^2\right]\EquationName{amgm-again}\\
{}&\le \prod_{t=1}^{CC(G'')} \sum_{a_1=1}^n \|u_{a_1}\|^2
\EquationName{useeven-again}\\
\nonumber {}&= \prod_{t=}^{CC(G'')} \|U\|_F^2\\
\nonumber {}&= d^{CC(G'')}
\end{align}
where \Equation{amgm-again} used the AM-GM inequality, and
\Equation{useeven-again} used \Lemma{evengraphs}, which applies since
$V(C_t)$ with edge set $T^t_1$ is connected, and $V(C_t)$ with edge
set $E(C_t)\backslash T^t_1$ is connected (since $T_2^t\subseteq
E(C_t)\backslash T^t_1$).
\end{proof}

Now, for any $G\in\mathcal{G}$ we have $y+z \le b+w$ since
for any graph the number of edges plus the number of connected
components is at
least the number of vertices. We also have $b \ge 2z$ since every
right vertex of $G$ is incident upon at least two distinct bonds
(since $i_t \neq j_t$ for all $t$). 
We also have $y\le b\le \ell$ since
$\MR(G)$ has exactly $2\ell$ edges with no isolated vertices,
and every bond has even multiplicity.
 Finally, a crude bound on the
number of different  $G\in\mathcal{G}$ with a given $b,y,z$ is
$(zy^2)^\ell \le (b^3)^{\ell}$. This is because when drawing the graph
edges in increasing order of edge label, when at a left vertex, we draw
edges from the left to the middle, then to the right, then to the
middle, and then
back to the left again, giving $y^2z$ choices. This is done $\ell$ times.
Thus by \Lemma{amgm} and \Equation{last-step}, and using that $t! \le
e\sqrt{t}(t/e)^t$ for all $t\ge 1$,
\allowdisplaybreaks
\begin{align}
\nonumber \E\tr((S-I)^\ell)&\le d\cdot
\frac{1}{s^\ell}\sum_{b,y,z,w}\sum_{\substack{G\in
    \mathcal{G}\\b(G)=b,y(G)=y\\w(G)=w,z(G)=z}}y!\cdot s^b\cdot
m^{z-b}\cdot d^{y-w}\\
\nonumber &\le ed\sqrt{\ell}\cdot
\frac{1}{s^\ell}\sum_{b,y,z,w}(b/e)^bs^b\sum_{\substack{G\in
    \mathcal{G}\\b(G)=b,y(G)=y\\w(G)=w,z(G)=z}}
\left(\frac dm\right)^{b-z}\\
\nonumber &\le ed\sqrt{\ell}\cdot
\frac{1}{s^\ell}\sum_{b,y,z,w}b^{3\ell}(b/e)^bs^b\cdot
\left(\frac dm\right)^{b-z}\\
\nonumber &\le ed\sqrt{\ell}\cdot
\frac{1}{s^\ell}\sum_{b,y,z,w}b^{3\ell}
\left((sb/e)\sqrt{\frac dm}\right)^b\\
&\le ed\ell^4\sqrt{\ell}\cdot \max_{2\le b\le \ell}
\left(\frac{b^3}{s}\right)^{\ell-b}
\left((b^4/e) \sqrt{\frac dm}\right)^b \EquationName{alternate-params}
\end{align}

Define $\epsilon = 2\eps - \eps^2$.
For $\ell \ge \ln(ed\ell^{9/2}/\delta) = O(\ln(d/\delta))$, $s \ge
e\ell^3/\epsilon = O(\log(d/\delta)^3/\eps)$, and $m \ge
d\ell^8/\epsilon^2 = O(d\log(d/\delta)^8/\eps^2)$, the above
expression is at
most $\delta \epsilon^\ell$. Thus as in \Equation{main-requirement},
by \Equation{trace-strategy2} we
have
$$ \Pr\left(\|S - I\| > \epsilon\right) < \frac 1{\epsilon^\ell}\cdot
\E \tr((S-I)^\ell) \le \delta .$$
\end{proof}

The proof of \Theorem{main-moments} reveals that for $\delta =
1/\mathrm{poly}(d)$ one could also set $m
= O(d^{1+\gamma}/\eps^2)$ and $s = O_{\gamma}(1/\eps)$ for any fixed
constant $\gamma>0$ and arrive at
the same conclusion. Indeed, let $\gamma' < \gamma$ be any positive
constant. Let $\ell$ in the proof
of \Theorem{main-moments} be taken as $O(\log (d/\delta)) = O(\log d)$.
It suffices to ensure  $\max_{2\le b\le
  \ell}(b^3/s)^{\ell-b}\cdot ((b^4/e) \sqrt{d/m})^b \le \eps^\ell
\delta / (e d\ell^{9/2})$ by \Equation{alternate-params}.
Note $d^{\gamma'} > b^{3\ell}$ as long as $b/\ln b >
3 \gamma^{-1}\ell/\ln d = O(1/\gamma')$, so $d^{\gamma'} >
b^{3\ell}$ for $b >
b^*$ for some $b^* = \Theta(\gamma^{-1}/\log(1/\gamma))$.
We choose $s \ge e (b^*)^3/\eps$ and $m =
d^{1+\gamma}/\eps^2$, which is at least
$d^{1+\gamma'}\ell^8/\eps^2$for $d$ larger than some
fixed constant. Thus the max above is always as small as desired,
which can be seen by looking at $b \le b^*$ and $b>b^*$ separately (in
the former case $b^3/s < 1/e$, and in the latter case 
$(b^3/s)^{\ell-b}\cdot ((b^4/e) \sqrt{d/m})^b < (\eps/e)^{\ell}
b^{3\ell} d^{-\gamma' b} = (\eps/e)^{\ell} e^{3\ell\ln b - \gamma'
  b\ln d} < (\eps/e)^{\ell}$ is as
small as desired). This observation yields:

\begin{theorem}\TheoremName{gamma}
Let $\alpha,\gamma>0$ be arbitrary constants. For $\SE$ an OSNAP with
$s=\Theta(1/\eps)$ and $\eps\in(0,1)$,
with probability at least $1-1/d^{\alpha}$, all singular
values of $\SE U$ are $1\pm\eps$
for $m = \Omega(d^{1+\gamma}/\eps^2)$
and $\sigma,h$ being $\Omega(\log d)$-wise independent. The constants
in
the big-$\Theta$ and big-$\Omega$ depend on $\alpha,\gamma$.
\end{theorem}

\begin{remark}
\textup{
\Section{intro} stated the time to list all non-zeroes
in a column in \Theorem{main-moments} is $t_c = \tilde{O}(s)$. For $\delta =
1/\mathrm{poly}(d)$, naively one would actually achieve $t_c =
O(s\cdot\log d)$ since one needs to evaluate an $O(\log d)$-wise
independent hash function $s$ times. This can be improved to
$\tilde{O}(s)$ using fast multipoint evaluation of hash functions; see
for example the last paragraph of Remark 16 of \cite{KNPW11}.
}
\end{remark}
\section{Applications}\SectionName{apps}
We use the fact that many matrix
problems  have the same time complexity as matrix multiplication
including computing the matrix inverse~\cite{BH74}\cite[Appendix
A]{HarveyThesis}, and QR
decomposition~\cite{Schonhage73}. In this paper we only consider the real RAM model and state the running time in terms of the number of field operations. The algorithms for solving linear systems, computing inverse, QR decomposition, and approximating SVD based on fast matrix multiplication can be implemented with precision comparable to that of conventional algorithms to achieve the same error bound (with a suitable notion of approximation/stability). We refer readers to~\cite{DDH07} for details. Notice that it is possible that both algorithms based on fast matrix multiplication and conventional counterparts are unstable, see e.g.~\cite{AV97} for an example of a pathological matrix with very high condition number.

In this section we describe some applications of our subspace
embeddings to problems in numerical linear algebra. All applications
follow from a straightforward replacement of previously used embeddings
with our new ones as most proofs go through verbatim. In the statement
of our bounds we implicitly assume $\nnz(A) \ge n$, since otherwise
fully zero rows of $A$ can be ignored without
affecting the problem solution.

\subsection{Approximate Leverage Scores}
This section describes the application of our subspace embedding from
\Theorem{main-moments} or \Theorem{gamma} to approximating the leverage
scores. Consider a matrix $A$ of size $n\times d$ and rank $r$. Let
$U$ be a $n\times r$ matrix whose columns form an orthonormal basis of
the column space of $A$. The {\em leverage scores} of $A$ are the
squared lengths of the rows of $U$. The algorithm for approximating
the leverage scores and the analysis are the same as those
of~\cite{CW12}, which itself uses essentially the same algorithm
outline as Algorithm 1 of \cite{DMMW12}. The improved bound is stated
below (cf.~\cite[Theorem 21]{CW12}).

\begin{theorem}
For any constant $\eps > 0$, there is an algorithm that with
probability at least $2/3$, approximates all leverage scores of a
$n\times d$ matrix $A$ in time $\tilde{O}(\nnz(A)/\eps^2 +
r^{\omega}\eps^{-2\omega})$.
\end{theorem}
\begin{proof}
As in \cite{CW12}, this follows by replacing the Fast
Johnson-Lindenstrauss embedding used in \cite{DMMW12} with our sparse
subspace embeddings. The only difference is in the parameters of our
OSNAPs. We essentially repeat the argument verbatim just to illustrate
where our new OSE parameters fit in; nothing in this proof is
new. Now, we first use
\cite{CKL12} so that we can assume $A$ has only $r=\rank(A)$ columns
and is of full column rank. Then, we take an OSNAP $\SE$ with $m =
\tilde{O}(r/\eps^2), s = (\polylog r)/\eps$ and compute $\SE A$. We
then find $R^{-1}$ so that $\SE A R^{-1}$ has orthonormal columns. The
analysis of \cite{DMMW12} shows that the $\ell_2^2$ of the rows of
$AR^{-1}$ are $1\pm\eps$ times the leverage scores of $A$. Take
$\SE'\in\R^{r\times t}$
to be a JL matrix that preserves the $\ell_2$ norms of the $n$ rows of
$AR^{-1}$ up to $1\pm\eps$. Finally, compute $R^{-1}\SE'$ then
$A(R^{-1}\SE')$ and output the squared row norms of $AR\SE'$.

Now we bound the running time. The time to reduce $A$ to having $r$
linearly independent columns
is $O((\nnz(A) + r^{\omega})\log n)$. $\SE A$ can be computed in
time $O(\nnz(A)\cdot (\polylog r)/\eps)$. Computing $R\in\R^{r\times
  r}$ from the $QR$ decomposition takes time $\tilde{O}(m^{\omega}) =
\tilde{O}(r^{\omega}/\eps^{2\omega})$, and then $R$ can be inverted in
time $\tilde{O}(r^{\omega})$; note $\SE AR^{-1}$ has orthonormal
columns. Computing $R^{-1}\SE'$ column by column takes time $O(r^2\log
r)$ using the FJLT of \cite{AL11,KW11} with $t =
O(\eps^{-2}\log n(\log\log n)^4)$. We then multiply the matrix $A$ by
the $r\times t$ matrix $R^{-1}\SE'$, which takes time $O(t\cdot
\nnz(A)) = \tilde{O}(\nnz(A)/\eps^2)$.
\end{proof}

\subsection{Least Squares Regression}

In this section, we describe the application of our subspace
embeddings to the problem of least squares regression. Here given a
matrix $A$ of size $n\times d$ and a vector $b\in \R^n$, the objective
is to find $x\in \R^d$ minimizing $\|Ax-b\|_2$. The reduction to
subspace embedding is similar to those of~\cite{CW12,Sarlos06}. The
proof is included for completeness.

\begin{theorem}
There is an algorithm for least squares regression running in time $O(\nnz(A) + d^3\log (d/\eps)/\eps^2)$ and succeeding with probability at least $2/3$.
\end{theorem}
\begin{proof}
Applying \Theorem{main} to the subspace spanned by columns of $A$ and
$b$, we get a distribution over matrices $\SE$ of size
$O(d^2/\eps^2)\times n$ such that $\SE$ preserves lengths of vectors
in the subspace up to a factor $1\pm\eps$ with probability at least
$5/6$. Thus, we only need to find $\argmin_x \|\SE Ax-\SE b\|_2$. Note
that $\SE A$ has size $O(d^2/\eps^2)\times d$. By Theorem 12
of~\cite{Sarlos06}, there is an algorithm that with probability at
least $5/6$, finds a $1\pm\eps$ approximate solution for least squares
regression for the smaller input of $\SE A$ and $\SE b$ and runs in time
$O(d^3\log(d/\eps)/\eps^2)$.
\end{proof}

The following theorem follows from using the embedding of \Theorem{main-moments} and the same argument as~\cite[Theorem 32]{CW12}.
\begin{theorem}
Let $r$ be the rank of $A$. There is an algorithm for least squares regression running in time $O(\nnz(A)((\log r)^{O(1)}+\log(n/\eps)) + r^{\omega}(\log r)^{O(1)} + r^2\log(1/\eps))$ and succeeding with probability at least $2/3$.
\end{theorem}

\subsection{Low Rank Approximation}
In this section, we describe the application of our subspace
embeddings to low rank approximation. Here given a matrix $A$, one
wants to find a rank $k$ matrix $A_k$ minimizing $\|A-A_k\|_F$. Let
$\Delta_k$ be the minimum $\|A-A_k\|_F$ over all rank $k$ matrices
$A_k$. Notice that our matrices are of the same form as sparse JL
matrices considered by~\cite{KN12} so the following property holds for
matrices constructed in \Theorem{main-moments} (cf.~\cite[Lemma
24]{CW12}).
\begin{theorem}~\cite[Theorem 19]{KN12}
Fix $\eps, \delta > 0$. Let $\mathcal{D}$ be the distribution over
matrices given in \Theorem{main-moments} with $n$ columns. For any
matrices $A, B$ with $n$ rows,
$$\Pr_{S\sim \mathcal{D}}[\|A^T S^T S B -A^TB\|_F > 3\eps/2\|A\|_F\|B\|_F] < \delta$$
\end{theorem}

The matrices of \Theorem{main} are the same as those of~\cite{CW12} so the above property holds for them as well. Therefore, the same algorithm and analysis as in~\cite{CW12} work. We state the improved bounds using the embedding of \Theorem{main} and \Theorem{main-moments} below (cf.~\cite[Theorem 36 and 38]{CW12}).
\begin{theorem}
Given a matrix $A$ of size $n\times n$, there are 2 algorithms that, with probability at least $3/5$, find 3 matrices $U,\Sigma, V$ where $U$ is of size $n\times k$, $\Sigma$ is of size $k \times k$, $V$ is of size $n\times k$, $U^T U = V^T V = I_k$, $\Sigma$ is a diagonal matrix, and
$$\|A-U\Sigma V^*\|_F\le (1+\eps)\Delta_k$$
The first algorithm runs in time $O(\nnz(A))+\tilde{O}(nk^2 + nk^{\omega-1}\eps^{-1-\omega} + k^{\omega}\eps^{-2-\omega})$. The second algorithm runs in time $O(\nnz(A)\log^{O(1)} k)+\tilde{O}(nk^{\omega-1}\eps^{-1-\omega} + k^{\omega}\eps^{-2-\omega})$.
\end{theorem}
\begin{proof}
The proof is essentially the same as that of~\cite{CW12} so we only
mention the difference. We use 2 bounds for the
running time: multiplying an $a\times b$ matrix and a $b\times c$
matrix with $c>a$ takes $O(a^{\omega-2}bc)$ time (simply dividing the
matrices into $a\times a$ blocks), and approximating SVD for an
$a\times b$ matrix $M$ with $a>b$ takes $O(ab^{\omega-1})$ time (time
to compute $M^T M$, approximate SVD of $M^T M = QDQ^T$ in
$O(b^{\omega})$ time~\cite{DDH07}, and compute $MQ$ to complete the
SVD of $M$).
\end{proof}

\section*{Acknowledgments}
We thank Andrew Drucker for suggesting the SNAP acronym for the OSE's
considered in this work, to which we added the ``oblivious''
descriptor.

\bibliographystyle{alpha}

\bibliography{../allpapers}

\begin{thebibliography}{DMIMW12}

\bibitem[AC09]{AC09}
Nir Ailon and Bernard Chazelle.
\newblock The {Fast} {Johnson--Lindenstrauss} transform and approximate nearest
  neighbors.
\newblock {\em SIAM J. Comput.}, 39(1):302--322, 2009.

\bibitem[Ach03]{Achlioptas03}
Dimitris Achlioptas.
\newblock Database-friendly random projections: {Johnson-Lindenstrauss} with
  binary coins.
\newblock {\em J. Comput. Syst. Sci.}, 66(4):671--687, 2003.

\bibitem[AL09]{AL09}
Nir Ailon and Edo Liberty.
\newblock Fast dimension reduction using {Rademacher} series on dual {BCH}
  codes.
\newblock {\em Discrete Comput. Geom.}, 42(4):615--630, 2009.

\bibitem[AL11]{AL11}
Nir Ailon and Edo Liberty.
\newblock Almost optimal unrestricted fast {Johnson-Lindenstrauss} transform.
\newblock In {\em Proceedings of the 22nd Annual ACM-SIAM Symposium on Discrete
  Algorithms (SODA)}, pages 185--191, 2011.

\bibitem[AV97]{AV97}
Noga Alon and Van~H. Vu.
\newblock Anti-{Hadamard} matrices, coin weighing, threshold gates, and
  indecomposable hypergraphs.
\newblock {\em J. Comb. Theory, Ser. A}, 79(1):133--160, 1997.

\bibitem[BH74]{BH74}
James~R. Bunch and John~E. Hopcroft.
\newblock Triangular factorization and inversion by fast matrix multiplication.
\newblock {\em Math. Comp.}, 28:231--236, 1974.

\bibitem[BOR10]{BOR10}
Vladimir Braverman, Rafail Ostrovsky, and Yuval Rabani.
\newblock Rademacher chaos, random {Eulerian} graphs and the sparse
  {Johnson-Lindenstrauss} transform.
\newblock {\em CoRR}, abs/1011.2590, 2010.

\bibitem[BY93]{BY93}
Z.D. Bai and Y.Q. Yin.
\newblock Limit of the smallest eigenvalue of a large dimensional sample
  covariance matrix.
\newblock {\em Ann. Probab.}, 21(3):1275--1294, 1993.

\bibitem[CW09]{CW09}
Kenneth~L. Clarkson and David~P. Woodruff.
\newblock Numerical linear algebra in the streaming model.
\newblock In {\em Proceedings of the 41st ACM Symposium on Theory of Computing
  (STOC)}, pages 205--214, 2009.

\bibitem[CW12]{CW12}
Kenneth~L. Clarkson and David~P. Woodruff.
\newblock Low rank approximation and regression in input sparsity time.
\newblock {\em CoRR}, abs/1207.6365v2, 2012.

\bibitem[DDH07]{DDH07}
James Demmel, Ioana Dumitriu, and Olga Holtz.
\newblock Fast linear algebra is stable.
\newblock {\em Numer. Math.}, 108(1):59--91, October 2007.

\bibitem[DKS10]{DKS10}
Anirban Dasgupta, Ravi Kumar, and Tam{\'a}s Sarl{\'o}s.
\newblock A sparse {Johnson-Lindenstrauss} transform.
\newblock In {\em Proceedings of the 42nd ACM Symposium on Theory of Computing
  (STOC)}, pages 341--350, 2010.

\bibitem[DMIMW12]{DMMW12}
Petros Drineas, Malik Magdon-Ismail, Michael Mahoney, and David Woodruff.
\newblock Fast approximation of matrix coherence and statistical leverage.
\newblock In {\em Proceedings of the 29th International Conference on Machine
  Learning (ICML)}, 2012.

\bibitem[Har08]{HarveyThesis}
Nicholas J.~A. Harvey.
\newblock {\em Matchings, Matroids and Submodular Functions}.
\newblock PhD thesis, Massachusetts Institute of Technology, 2008.

\bibitem[HMT11]{HMT11}
Nathan Halko, Per-Gunnar Martinsson, and Joel~A. Tropp.
\newblock Finding structure with randomness: Probabilistic algorithms for
  constructing approximate matrix decompositions.
\newblock {\em SIAM Rev., Survey and Review section}, 53(2):217--288, 2011.

\bibitem[HV11]{HV11}
Aicke Hinrichs and Jan Vyb\'{\i}ral.
\newblock Johnson-lindenstrauss lemma for circulant matrices.
\newblock {\em Random Struct. Algorithms}, 39(3):391--398, 2011.

\bibitem[HW71]{HW71}
David~Lee Hanson and Farroll~Tim Wright.
\newblock A bound on tail probabilities for quadratic forms in independent
  random variables.
\newblock {\em Ann. Math. Statist.}, 42(3):1079--1083, 1971.

\bibitem[JL84]{JL84}
William~B. Johnson and Joram Lindenstrauss.
\newblock Extensions of {Lipschitz} mappings into a {Hilbert} space.
\newblock {\em Contemporary Mathematics}, 26:189--206, 1984.

\bibitem[KMR12]{KMR12}
Felix Krahmer, Shahar Mendelson, and Holger Rauhut.
\newblock Suprema of chaos processes and the restricted isometry property.
\newblock {\em arXiv}, abs/1207.0235, 2012.

\bibitem[KN10]{KN10}
Daniel~M. Kane and Jelani Nelson.
\newblock A derandomized sparse {Johnson-Lindenstrauss} transform.
\newblock {\em CoRR}, abs/1006.3585, 2010.

\bibitem[KN12]{KN12}
Daniel~M. Kane and Jelani Nelson.
\newblock Sparser {Johnson}-{Lindenstrauss} transforms.
\newblock In {\em SODA}, pages 1195--1206, 2012.

\bibitem[KNPW11]{KNPW11}
Daniel~M. Kane, Jelani Nelson, Ely Porat, and David~P. Woodruff.
\newblock Fast moment estimation in data streams in optimal space.
\newblock In {\em Proceedings of the 43rd ACM Symposium on Theory of Computing
  (STOC)}, pages 745--754, 2011.

\bibitem[KW11]{KW11}
Felix Krahmer and Rachel Ward.
\newblock New and improved {J}ohnson-{L}indenstrauss embeddings via the
  {R}estricted {I}sometry {P}roperty.
\newblock {\em SIAM J. Math. Anal.}, 43(3):1269--1281, 2011.

\bibitem[Mah11]{Mahoney11}
Michael~W. Mahoney.
\newblock Randomized algorithms for matrices and data.
\newblock {\em Foundations and Trends in Machine Learning}, 3(2):123--224,
  2011.

\bibitem[NDT09]{NDT09}
Nam~H. Nguyen, Thong~T. Do, and Trac~D. Tran.
\newblock A fast and efficient algorithm for low-rank approximation of a
  matrix.
\newblock In {\em Proceedings of the 41st ACM Symposium on Theory of Computing
  (STOC)}, pages 215--224, 2009.

\bibitem[NN12]{NN12b}
Jelani Nelson and Huy~L. Nguy$\tilde{\hat{\mbox{e}}}$n.
\newblock Sparsity lower bounds for dimensionality-reducing maps.
\newblock Manuscript, 2012.

\bibitem[NW61]{NashWilliams61}
Crispin St. John~Alvah Nash-Williams.
\newblock Edge-disjoint spanning trees of finite graphs.
\newblock {\em J. London Math. Soc.}, 36:445--450, 1961.

\bibitem[Sar06]{Sarlos06}
Tam{\'a}s Sarl{\'o}s.
\newblock Improved approximation algorithms for large matrices via random
  projections.
\newblock In {\em Proceedings of the 47th Annual IEEE Symposium on Foundations
  of Computer Science (FOCS)}, pages 143--152, 2006.

\bibitem[Sch73]{Schonhage73}
Arnold Sch\"onhage.
\newblock Unit\"are transformationen gro\ss{}er matrizen.
\newblock {\em Numer. Math.}, 20:409--417, 1973.

\bibitem[Tao12]{Tao12}
Terence Tao.
\newblock {\em Topics in random matrix theory}, volume 132 of {\em Graduate
  Studies in Mathematics}.
\newblock American Mathematical Society, 2012.

\bibitem[Tro11]{Tropp11}
Joel~A. Tropp.
\newblock Improved analysis of the subsampled randomized {Hadamard} transform.
\newblock {\em Adv. Adapt. Data Anal., Special Issue on Sparse Representation
  of Data and Images}, 3(1--2):115--126, 2011.

\bibitem[Tut61]{Tutte61}
William~Thomas Tutte.
\newblock On the problem of decomposing a graph into $n$ connected factors.
\newblock {\em J. London Math. Soc.}, 142:221--230, 1961.

\bibitem[TZ12]{TZ12}
Mikkel Thorup and Yin Zhang.
\newblock Tabulation-based 5-independent hashing with applications to linear
  probing and second moment estimation.
\newblock {\em SIAM J. Comput.}, 41(2):293--331, 2012.

\bibitem[Ver12]{Vershynin12}
Roman Vershynin.
\newblock Introduction to the non-asymptotic analysis of random matrices.
\newblock In Y.~Eldar and G.~Kutyniok, editors, {\em Compressed Sensing, Theory
  and Applications}, chapter~5, pages 210--268. Cambridge University Press,
  2012.

\bibitem[Vyb11]{Vybiral11}
Jan Vyb\'{i}ral.
\newblock A variant of the {Johnson-Lindenstrauss} lemma for circulant
  matrices.
\newblock {\em J. Funct. Anal.}, 260(4):1096--1105, 2011.

\bibitem[Wil12]{Williams12}
Virginia~Vassilevska Williams.
\newblock Multiplying matrices faster than {Coppersmith}-{Winograd}.
\newblock In {\em STOC}, pages 887--898, 2012.

\bibitem[yCKL12]{CKL12}
Ho~yee Cheung, Tsz~Chiu Kwok, and Lap~Chi Lau.
\newblock Fast matrix rank algorithms and applications.
\newblock In {\em Proceedings of the 44th Symposium on Theory of Computing
  (STOC)}, pages 549--562, 2012.

\bibitem[ZWSP08]{ZWSP08}
Yunhong Zhou, Dennis~M. Wilkinson, Robert Schreiber, and Rong Pan.
\newblock Large-scale parallel collaborative filtering for the netflix prize.
\newblock In {\em Proceedings of the 4th International Conference on
  Algorithmic Aspects in Information and Management (AAIM)}, pages 337--348,
  2008.

\end{thebibliography}

\end{document}